\documentclass[runningheads]{llncs}
\usepackage{graphicx}
\usepackage{hyperref}

\usepackage{verbatim}
\usepackage{multicol}
\usepackage{tikz}

\usepackage{subfigure}

\usepackage{amsmath}
\usepackage{amsfonts}
\usepackage{amssymb}
\usepackage{mathrsfs}

\usetikzlibrary{arrows,backgrounds,calc,fit,decorations.pathreplacing,decorations.markings,shapes.geometric, arrows}
\tikzset{every fit/.append style=text badly centered}
\definecolor{codegreen}{rgb}{0,0.6,0}
\definecolor{codegray}{rgb}{0.5,0.5,0.5}
\definecolor{codepurple}{rgb}{0.58,0,0.82}
\definecolor{backcolour}{rgb}{0.95,0.95,0.92}

\tikzstyle{internal} = [draw, fill, shape=circle]
\tikzstyle{external} = [shape=circle]
\tikzstyle{square}   = [draw, fill, rectangle]
\tikzstyle{triangle} = [draw, fill, regular polygon, regular polygon sides=3, inner sep=3pt]
\tikzstyle{pentagon} = [draw, fill, regular polygon, regular polygon sides=5, inner sep=2pt, minimum size=14pt]

\newcommand{\Rmnum}[1]{\expandafter\@slowromancap\romannumeral #1@}
\makeatother

\newcommand{\Holant}{\operatorname{Holant}}

\newcommand{\holant}[2]{\ensuremath{\Holant\left(#1\mid #2\right)}}

\newcommand{\mf}[1]{\ensuremath{\mathcal{#1}}}
\newcommand{\mb}[1]{\ensuremath{\mathbb{#1}}}

\newcommand{\numP}{{\rm \#P}}

\newcommand{\innerp}[2]{\ensuremath{\la #1, #2 \ra }}

\newcommand{\la}{\langle}
\newcommand{\ra}{\rangle}

\newcommand{\tent}[2]{\ensuremath{#1 ^ {\otimes #2}}}
\newcommand{\tprime}[2]{\ensuremath{#1 ^ {\prime\otimes #2}}}

\begin{document}
\title{{ Restricted Holant
 Dichotomy on Domains 3 and  4}}


%
%
\author{Yin Liu \and
Austen Z. Fan \and
Jin-Yi Cai}
\authorrunning{Y. Liu et al.}
\institute{University of Wisconsin-Madison, Madison WI 53715, USA \\
\email{\{yinl,afan,jyc\}@cs.wisc.edu} 
}



%
%
\maketitle 
\begin{abstract}
 $\Holant^*(f)$ denotes a
 class of counting problems specified by a constraint function $f$.  We prove complexity dichotomy theorems for  $\Holant^*(f)$ in two settings:
(1) $f$ is any arity-3 real-valued function on input of
domain size 3. (2) 
 $f$ is any arity-3 $\{0,1\}$-valued function on input of
domain size 4.
\end{abstract}

\section{Introduction and background}
Counting problems arise in many branches in 
computer science, machine learning  and statistical physics. Holant problems encompass a broad class of
counting problems \cite{Backens21,BackensG20,CaiGW16,CaiL11,cai2009holant,CaiLX10,GuoHLX11,GuoLV13,MKJYC,valiant2006accidental,valiant2008holographic,Xia11}. For symmetric constraint functions (a.k.a. signatures) this is also equivalent to
edge-coloring models~\cite{szegedy2007edge,szegedy2010edge}.
These problems extend counting constraint satisfaction problems. Freedman, Lov\'{a}sz and Schrijver proved that some prototypical Holant problems, such as counting perfect matchings, cannot be expressed as vertex-coloring models known as graph homomorphisms~\cite{Freedman-Lovasz-Schrijver-2007,HellN04}. 
The complexity classification program of counting problems 
is to classify the computational complexity
of these problems.

Formally, a Holant problem on domain $D$ is defined on a graph $G=(V, E)$ where edges are
variables and vertices are constraint functions.
Given a set of  constraint functions $\mathcal{F}$ defined on $D$,
a \emph{signature grid} $\Omega=(G, \pi)$  assigns to each vertex $v \in V$ an $f_{v} \in \mathcal{F}$.
The aim is to compute the following partition function 
$$\Holant_\Omega = \sum_{\sigma: E \rightarrow D} \prod_{v \in V}f_v\left(\left.\sigma\right|_{E(v)}\right).$$
The computational problem is denoted by $\Holant (\mf{F})$.
E.g., on the Boolean domain, it is over all $\{0,1\}$-edge assignments. On domain size 3, it is over all $\{R,G,B\}$-edge assignments, signifying three colors Red, Green and Blue. On domain size 4, it is over all $\{R,G,B,W\}$-edge assignments.
On the Boolean domain, if every vertex has the 
\textsc{Exact-One}
function (which evaluates to 1 if
exactly one incident edge is 1, and evaluates to  0 otherwise),
then the partition function gives
the number of perfect matchings.
On domain size $k$, if every vertex has the 
\textsc{All-Distinct}
function,
then the partition function gives
the number of valid $k$-edge colorings.

A \emph{symmetric} signature is a function that is invariant under
any permutation of its variables.
The value of such a signature  depends only on the numbers of each color assigned to its input variables. 
The number of variables
is its arity; unary, binary, ternary
signatures have arities 1, 2, 3.
We denote a symmetric ternary signature $g$ on domain size 3 by a ``triangle" consisting of 10 numbers:
\begin{center}
\begin{tabular}{c c c c c c c}
     &&& $g_{3,0,0}$ &&&  \\
     && $g_{2,1,0}$ && $g_{2,0,1}$ &&\\
   & $g_{1,2,0}$ && $g_{1,1,1}$ && $g_{1,0,2}$ & \\
   $g_{0,3,0}$ && $g_{0,2,1}$ && $g_{0,1,2}$ && $g_{0,0,3}$ 
\end{tabular}    
\end{center}
where $g_{i,j,k}$ is the value on inputs having $i$ Red, $j$ Green and $k$ Blue.
Similarly, we denote a symmetric ternary signature $g$ on domain size 4 by a ``tetrahedron":

\begin{figure}[h!]
    \centering
    \includegraphics[scale=0.6]{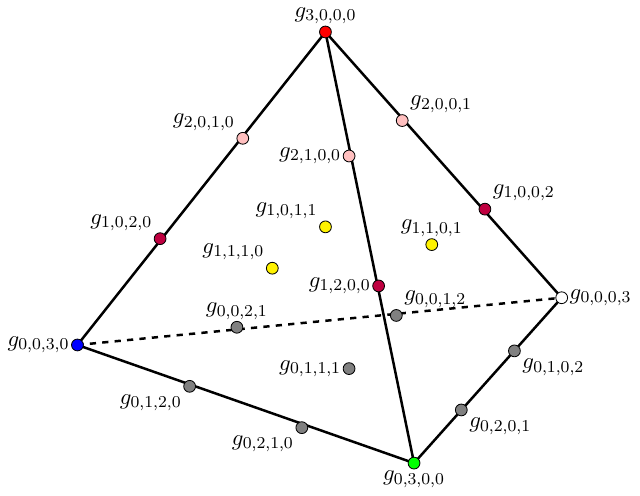}
    \label{pyramid}
\end{figure}

While much progress has been made for the
classification of counting CSP~\cite{bulatov2006dichotomy}  
\cite{jacm/CaiC17,cai2016nonnegative,dyer2013effective}, and some progress for
Holant problems~\cite{CaiGW14,CaiLX13,Fu-Yang-Yin}, classifying Holant problems on higher domains is particularly challenging.
One of the few existing work on a higher domain is \cite{CaiLX13}, in which a dichotomy for $\Holant^*(f)$ is proved where $f$ is a ternary complex symmetric function on domain size 3 and the $*$ means all unary functions are available. (Note that Holant problems with signatures
of arity $\le 2$ are all P-time tractable;  the interesting case where both
tractable and \#P-hardness occur starts with ternary signatures.)

In this work, we attempt to extend this to
Holant problems on domain size 4.
Our effort only met with partial success.
We are able to prove a
complexity dichotomy for $\Holant^*(f)$ for any $\{0,1\}$-valued symmetric ternary constraint function $f$ (see Theorem~\ref{main}). 

Our technique is to try to reduce a domain 4
problem to a domain 3 problem, and then analyze
the situation using the existing domain 3
dichotomy~\cite{CaiLX13}. To do so, we
will need to be able to construct (or interpolate) a suitable constraint function that allows us to effectively restrict the problem
to a domain 3 problem, where the new domain elements
are superpositions of old domain elements
under a holographic transformation. This turns out
to be a nontrivial task. 
And one reason that we cannot extend to a more general domain 4 dichotomy is that for some real-valued signatures, it is impossible to construct such a constraint function. 
On the other hand, for $\{0,1\}$-valued domain 4 signatures, we are able to succeed in this
plan (using several different constructions).

To carry out this plan, we use the 
domain 3
dichotomy~\cite{CaiLX13} extensively.
This motivates us to examine the domain 3
dichotomy more closely,
 when the constraint function $f$
is real-valued. 
Since the
domain 3
dichotomy~\cite{CaiLX13} applies to
all complex-valued functions it certainly
also applies to real-valued functions\footnote{There is
a slight issue that for a real-valued $f$,   $\Holant^*(f)$ naturally refers to having free  real-valued unary functions, while the existing 
 $\Holant^*$ dichotomy for complex-valued $f$ assumes all complex-valued unary functions are available for free.
In Lemma~\ref{star} we address this
technical difficulty.}. 
However, we found out that applying the
dichotomy for complex-valued functions directly
is very cumbersome,
so much so that the attempt to use it
for our exploration in domain 4 grinds to a halt.

So, we return to domain size 3, and found
that there is a cleaner form of the
dichotomy  of $\Holant^*(f)$ 
on domain size 3 for real-valued $f$.
This turns out to be 
a non-trivial adaptation, as
we prove that certain tractable forms
in the complex case \emph{cannot} occur
for real-valued signatures (see Theorem~\ref{co1}). In the proof of this real $\Holant^*$  dichotomy, orthogonal holographic transformations are heavily used.


%
Armed with this more effective form of
the domain 3 dichotomy, we return to  Holant problems on domain 4, and prove a dichotomy for $\Holant^*(f)$ where $f$ is $\{0,1\}$-valued (see Theorem~\ref{main}). 
We use several  
strategies that are more generally applicable and are  worth mentioning (see subsection~\ref{appro}).
For general real-valued ternary symmetric signatures on domain 4, we prove Theorem~\ref{d4tra}. It gives some broad classes of constraint functions that define P-time tractable Holant problems. We conjecture
that this is actually a complexity dichotomy. 


\vspace{-.1in}
\paragraph{\bf Some preliminaries}
We can picture a signature as a vertex with several dangling edges as its input variables. Connecting a unary signature $u$
to another signature $f$ of arity $r\ge 1$ creates
a signature of arity $r-1$. 
If $f$ is symmetric then this does not
depend on which variable (dangling edge of $f$) $u$
is connected to, and the resulting
signature is denoted by $\la f, u \ra$.
In particular, if $f$ is also
a unary, then $\la f, u \ra$
is a scalar value, equal to
the dot product of their signature entries.
One should note that, for
complex unary signatures, this dot product 
(without conjugation) is
not the usual inner product 
and it is possible that $\la u, u \ra =0$
for $u \ne 0$.
We call a vector $u$  \emph{isotropic} if
$\la u, u \ra =0$ (including $u=0$).

We now introduce  \emph{holographic transformations}. 
We use $\holant{\mathcal{R}}{\mathcal{G}}$ to denote bipartite Holant problems on bipartite graphs $H=(U,V,E)$, where each signature on a vertex in $U$ or $V$ is from $\mf{R}$  or  $\mf{G}$, respectively.
Suppose 
$T$ is an invertible matrix of the same size  as that of the domain.
We say that there is a holographic transformation from $\holant{\mf{R}}{\mf{G}}$ to $\holant{\mf{R}'}{\mf{G}'}$ by $T$,
if $\mf{R}' = \mf{R}T^{-1}$ 
and
$\mf{G}' = T\mf{G}$, where $\mf{R}T^{-1}=\{f\tent{(T^{-1})}{r(f)} \mid f \in \mathcal{R}\}$,  
$T\mf{G}=\{\tent{T}{r(f)}f \mid f \in \mathcal{G}\}$
and   $r(f)$ is the arity of $f$. 
Here each signature  is written
as a column/row vector in lexicographical order as a truth-table.
We also write $Tf$ for $\tent{T}{r(f)}f$ when the arity  $r(f)$ is clear.

\begin{theorem}[Valiant's Holant Theorem~\cite{valiant2008holographic}]
Suppose there  is a holographic transformation from $\holant{\mf{R}}{\mf{G}}$ to $\holant{\mf{R}'}{\mf{G}'}$, then $\holant{\mf{R}}{\mf{G}}$ $ \equiv_T \holant{\mf{R}'}{\mf{G}'}$, where $\equiv_T$ means equivalence up to a P-time reduction.
\end{theorem}

Therefore, if there is a holographic transformation from $\holant{\mf{G}}{\mf{R}}$ to $\holant{\mf{G}'}{\mf{R}'}$, then one problem is in P iff the other one is, and similarly one problem is \numP-hard iff the other one is. For any general graph, we can make it bipartite by adding an additional vertex on each edge (thus forming the vertex-edge incidence graph), and assigning those new vertices the binary \textsc{Equality} signature $(=_2)$  of the corresponding domain size.
Note that for the binary \textsc{Equality},
 if $T$ is  orthogonal, then 
 it is unchanged under the holographic transformation by $T$.
 Hence, $\Holant(\mf{F}) \equiv_T \holant{=_2}{\mf{F}} \equiv_T 
\Holant(T\mf{F})$ for any orthogonal matrix $T$.

\section{A real dichotomy for $\Holant^*(f)$ on domain 3}
In this section we prove a complexity
dichotomy of $\Holant^*(f)$ for any
\emph{real-valued} 
symmetric ternary function $f$ over the domain $\{R,G,B\}$. We investigate
the three tractable forms of 
Theorems 3.1 and 3.2 in \cite{CaiLX13}
when $f$ is real-valued.  It turns out that they take more special forms,
and one  tractable case
for complex-valued $f$ does not occur
for real-valued  $f$. However,
complex tensors are still needed to express one of the two tractable forms. Lemmas~\ref{rt1},~\ref{rt2} and~\ref{rt3}
address each of the three tractable families.

\begin{lemma} \label{ct0}
For all $\beta \in \mathbb{C}^3$,  if $\la\beta, \beta\ra =0$ then there exists a 3-by-3 real orthogonal matrix $T$, such that $T\beta = c (1, i, 0) ^T$ where $c\in \mathbb{R}$.
\end{lemma}

\begin{proof}
Write $\beta = \gamma + \delta i$, $\gamma, \delta \in \mathbb{R}^3$. 
Then $ 0 = \la\beta, \beta\ra = \la\gamma,\gamma\ra - \la \delta ,\delta \ra + 2\la \gamma, \delta \ra i $.
Considering its real and imaginary parts separately,  $ \lVert \gamma \rVert  = \lVert \delta \rVert $ and $ \gamma \perp \delta $.  Then  there exists a 
real orthogonal $T$, 
such that
  $T\gamma = ce_1$ and $
    T\delta = c e_2$, 
where  $c = \lVert \gamma \rVert  \in \mathbb{R}$. It follows that
$  T\beta = T(\gamma + \delta i)  = c(1,i,0)^T$.
\end{proof}



\begin{lemma}  \label{rt1}
If there exist $\alpha, \beta, \gamma$ $\in$ $\mathbb{C}^3$, $s.t., f=  \tent{\alpha}{3}+ \tent{\beta}{3} + \tent{\gamma}{3}$, $\la\alpha,\beta\ra = \la\beta,\gamma\ra = \la\gamma, \alpha\ra = 0$,
and $f$ is real-valued, then there exist $\alpha', \beta', \gamma' \in \mathbb{R}^3, s.t., f=  \tprime{\alpha}{3}+ \tprime{\beta}{3} + \tprime{\gamma}{3}, \la\alpha',\beta'\ra = \la\beta',\gamma'\ra = \la\gamma', \alpha'\ra = 0$. Thus, there is a real orthogonal transformation $T$, such that $Tf = a  \tent{e_1}{3}+ b\tent{e_2}{3} +c \tent{e_3}{3}$, for some $a, b, c\in \mathbb{R}$.
\end{lemma}

\begin{proof}
    Let $M_i = \la f, e_i \ra = \alpha_i \tent{\alpha}{2} + \beta_i \tent{\beta}{2} + \gamma_i \tent{\gamma}{2}, i = 1,2,3$, then $M_i$ is a real symmetric matrix.

    If there is any  $v$ among $\{\alpha, \beta, \gamma\}$ that is non-isotropic, i.e., $\la v, v \ra \ne 0$, then by symmetry,
    assume it is $\alpha$. Then $\alpha \ne 0$. At least one of $\alpha_i \ne 0$.
    Say $\alpha_1 \ne 0$.

    We have $M_i\alpha = \lambda_i \alpha$, where
    $\lambda_i  = \alpha_i \la\alpha, \alpha \ra$.  So $\lambda_i$ is an eigenvalue of a real symmetric matrix $M_i$, therefore it is real ($i=1,2,3$). 
    As $\la\alpha, \alpha \ra \ne 0$ and $\alpha_1 \ne 0$, 
    $\frac{\alpha_i}{\alpha_1} = \frac{\lambda_i}{\lambda_1}$
    is real and is well defined ($i = 1,2,3$).
    We can then write $\alpha = \mu u$, where $\mu \in \mathbb{C}, u \in \mathbb{R}^3$ and $\lVert u \rVert = 1$. As $\alpha \ne 0$,
    we have $\mu \ne 0$ and $\la u, \beta \ra  = \la u, \gamma \ra = 0$.
    
    Thus $ f = \mu^3 \tent{ u}{3}+ \tent{\beta}{3} + \tent{\gamma}{3}$. 
    We have $ \la f, u\ra = \mu^3 \tent{u}{2}$. Since $f$ and $u$ are both real, and $u \ne 0$, we have $\mu^3 \in \mathbb{R}$.
    Thus, $\mu^3 = t^3$
    for some real $t\in \mathbb{R}$. It follows that
    $\tent{\alpha}{3} = t^3\tent{u}{3} = \tent{(tu)}{3}$.
    
    Then we replace $\alpha$ with $tu$. 
    Similarly, if $\beta$ or $\gamma$ is a non-isotropic vector, we can replace it with a real vector, without changing $f$. 
    Thus we get a new form $f = \tent{\alpha}{3}+ \tent{\beta}{3} + \tent{\gamma}{3}$ where $\alpha, \beta, \gamma$ are either real or isotropic
    (and since the zero vector
    is real we may further assume $\alpha, \beta, \gamma$ are either real or nonzero isotropic.)

    If they are all real, then we can use a real orthogonal matrix $T$ to transform $f$, i.e., $Tf =a  \tent{e_1}{3}+ b\tent{e_2}{3} +c \tent{e_3}{3}$, for some $a,b,c \in \mathbb{R}$. Then we are done. 

    Now suppose there is at least one nonzero isotropic vector among $\{\alpha, \beta, \gamma\}$. W.o.l.o.g., we can assume $\gamma$ is nonzero and isotropic. By Lemma~\ref{ct0}, there exists a real orthogonal $T$, s.t., $T\gamma = r(1,i,0)^T, r \in \mathbb{R}\setminus\{0\}$ . 
    Then since $\la T\alpha, T\gamma \ra = \la T\beta, T\gamma \ra = 0$, $T\alpha, T\beta$ each must have the form  $(c, ci, d)^T$. 
    If, in addition, $\alpha$ is also isotropic, then $T\alpha$ must have the form $(c,ci,0)^T$, which is a multiple of $T\gamma$.
    As $\gamma$ is nonzero, $\alpha$ is hence a multiple of $\gamma$.
    Then $\alpha$ can be absorbed into $\gamma$ and form a new isotropic vector $\gamma'$, and $\la \gamma', \beta \ra =0$, $f = \tent{\beta}{3} + \tprime{\gamma}{3}$. 
    We have the same argument for $\beta$. So w.o.l.o.g., we can write $f = \tent{\alpha}{3}+ \tent{\beta}{3} + \tent{\gamma}{3} $, and there is \emph{at most one},
    and therefore, \emph{exactly one}, nonzero isotropic vector among $\{\alpha, \beta, \gamma\}$, and the others are real vectors and could be zero. 
   We have
    $f = \tent{\gamma}{3} + R$, where $R$ is some real-valued tensor. 
    Hence $\tent{\gamma}{3}$ is real.
Since $\gamma \ne 0$,
we have $\tent{\gamma}{3} \ne 0$.
But then, $\la \la \la \tent{\gamma}{3},
\gamma \ra, \gamma \ra, \gamma \ra= \la \tent{\gamma}{3},
\tent{\gamma}{3} \ra \ne 0$.
In particular,
$0 \ne \la \tent{\gamma}{3},
\gamma \ra =
\la \gamma,
\gamma \ra \tent{\gamma}{2}
= 0$. 
This is a contradiction.  
    
\end{proof}

\begin{lemma} \label{rt2}
 If there exist $\alpha, \beta_1, \beta_2$ $\in \mathbb{C}^3, s.t., f = \tent{\alpha}{3}+ \tent{\beta_1}{3} + \tent{\beta_2}{3}, \la\alpha, \beta_i\ra = \la\beta_i, \beta_i \ra = 0, i = 1,2$, and $f$ is real-valued,  then there is a real orthogonal transformation $T$, such that $cTf = \epsilon(\tent{\beta_0}{3} + \tent{\overline{\beta_0}}{3}) + \lambda\tent{e_3}{3}$, where $\beta_0 = (1,i,0)^T, \epsilon \in \{0,1\}$ and $c, \lambda \in \mathbb{R}, c \ne 0$. Thus,  there exist $\alpha \in \mathbb{R}^3$ and $\beta \in \mb{C}^3, s.t., f = \tent{\alpha}{3}+ \tent{\beta}{3} + \tent{\overline{\beta}}{3}$,  $ \la\alpha, \beta\ra = \la\beta, \beta \ra = 0$. 
    
\end{lemma}

\begin{proof}
    First we assume $\beta_1 = \beta_2 =0$. As $f$ is real, by a similar argument as in the proof of Lemma~\ref{rt1}, we know that there exists an $\alpha' \in \mathbb{R}^3, f=  \tent{\alpha}{3} = \tprime{\alpha}{3}$. Thus there is a real orthogonal $T$, $T\alpha' = (0,0,d)^T, d\in \mathbb{R}$, so $Tf = \tent{(T\alpha')}{3}  = d^3 \tent{e_3}{3}$.

    Now w.o.l.o.g., we can assume $\beta_1 \ne 0$. 
    
    Suppose $\la\alpha, \alpha \ra = 0$.
    By Lemma~\ref{ct0}, there exists a real orthogonal transformation $T$, such that $T\beta_1=c(1,i,0)^T$ 
    for some $c \ne0$. Since  $\la T\alpha, T\beta_1 \ra =0$, and that $\alpha$ is also isotropic,  $T\alpha$ must have the form  $(a,ai,0)^T$, thus a multiple of $T\beta_1$.
    So, $\alpha$  is a multiple of $\beta_1$. Then $\alpha$ can be absorbed into $\beta_1$ and form a new isotropic vector $\mu\beta_1$ for some $\mu$.
    Then we replace $\beta_1$ with $\mu \beta_1$ and $f = \tent{ \beta_1}{3} + \tent{\beta_2}{3}$. 

    Else, $\la \alpha, \alpha \ra \ne 0$, then $\alpha \ne 0$. Let $M_i = \la f, e_i\ra, i = 1,2,3$, then $M_i$ is a real symmetric matrix. $M_i \alpha = \alpha_i\la \alpha, \alpha\ra \alpha$. 
    By the same argument as in the proof of Lemma~\ref{rt1},  there exists an $\alpha'\in \mathbb{R}^3, s.t., \tent{\alpha}{3} = \tprime{\alpha}{3}$. 

    Hence in both cases, we get a new form $f = \tent{\alpha}{3}+ \tent{\beta_1}{3} + \tent{\beta_2}{3}$, where $\alpha \in \mathbb{R}^3$ is real (possibly 0), $\la\alpha, \beta_i\ra = \la\beta_i, \beta_i \ra = 0, i = 1,2$. 
    Since $\alpha$ is real, there is some real orthogonal $T$, $T\alpha = te_3 = (0,0,t)^T, t\in \mathbb{R}$. If $\beta_1 = \beta_2 = 0$, we are done. Thus we may assume $\beta_1 \ne0$. 

\vspace{.1in}
    \noindent{\bf Case 1:}
$\alpha \ne 0$, i.e.,  $t\ne 0$. 
\nopagebreak
    Since $\la T\beta_i, T\alpha \ra = 0, i=1,2$ and $\beta_i$ is isotropic, we have $T\beta_1 = u(1, \pm i, 0)^T$, $T\beta_2 = v(1,\pm i, 0)^T, u\ne 0$. 
    If $\beta_2$ is a multiple of $\beta_1$, it can be absorbed into $\beta_1$ and form a new isotropic vector $s(1, \pm i, 0)^T$ for some $s$. Then we get $Tf = \tent{(te_3)}{3} + s^3\tent{((1,\pm i, 0)^T)}{3}$.
    As $Tf$ is real and $t$ is real, we get $s=0$ and we are done.
    Else, $\beta_1, \beta_2$ are independent, i.e., $Tf = \tent{(te_3)}{3} + u^3\tent{((1,\pm i,0)^T)}{3} + v^3\tent{((1, \mp i,0)^T)}{3}$ 
    where $uv\ne0$. 
    As $Tf$ is real, it follows that 
    $u^3 + v^3 \in \mathbb{R}$
    and
    $u^3 - v^3 =0$, and
    thus 
   $u^3 = v^3 \in \mathbb{R}$.
    So $\frac{1}{u^3}Tf =\tent{\beta_0}{3} + \tent{\overline{\beta_0}}{3} + \frac{t^3}{u^3}\tent{e_3}{3}$, where $\beta_0 = (1,i,0)^T$.

\vspace{.05in}
\noindent{\bf Case 2:}
$\alpha = 0$, i.e.,  $t = 0$. 
\nopagebreak
    We have $f = \tent{\beta_1}{3} + \tent{\beta_2}{3}$, $\la\beta_i, \beta_i\ra = 0$, $i=1,2$, 
    and $\beta_1\ne 0$. 
    By Lemma~\ref{ct0}, there exists a real orthogonal $T$, $T\beta_1 = u(1,i,0)^T, u \in \mathbb{R} \setminus \{0\}$. 
    If $\beta_2 =0$, then $Tf = \tent{T\beta_1}{3} = u^3\tent{((1,i,0)^T)}{3}$, where the LHS is real but the RHS is not,
    which is a contradiction. 
    So we have $\beta_2 \ne 0$.
    Let $\beta_i' = T\beta_i$, and we have $Tf = \tprime{\beta_1}{3} + \tprime{\beta_2}{3}$.
    Let $M_i = \la Tf, e_i\ra = \beta_{1i}'\tprime{\beta_1}{2} + \beta_{2i}'\tprime{\beta_2}{2}, i = 1,2,3$. Both $\beta_1',
    \beta_2' \ne 0$.
    Then,
    \[
    \begin{cases}
        M_i\beta_1' = \lambda_{2i}\beta_2' \\
        M_i\beta_2' = \lambda_{1i}\beta_1',
    \end{cases}
\mbox{ where } \begin{cases}
        \lambda_{1i} = \beta_{1i}'\la \beta_1', \beta_2' \ra \\
        \lambda_{2i} = \beta_{2i}'\la\beta_1', \beta_2'\ra .
    \end{cases}
\]
    
Applying $M_i$ twice, we get  $ M_i^2\beta_1' = \lambda_{2i}\lambda_{1i} \beta_1'$.    
    Since $M_i$ is real symmetric, so is $M_i^2$.
    As $\beta_1' \ne 0$,
  $ \lambda_{2i}\lambda_{1i} $ is an eigenvalue of $M_i^2$, therefore real, i.e., $\beta_{1i}'\beta_{2i}'\la\beta_1', \beta_2'\ra^2$ is real, $i=1,2,3$. 
    Recall $\beta_1' = u(1,i,0)^T$, and now let $\beta_2' = (x,y,z)^T \in \mathbb{C}^3$. Let
    $\tau= x+yi$, and $\mu = \la\beta_1'$, $\beta_2'\ra = u\tau$.
    Then we have $ \begin{cases}
        \beta_{11}'\beta_{21}'\mu^2 = ux \cdot u^2\tau^2 = u^3\tau^2x \\
        \beta_{12}'\beta_{22}'\mu^2 = uiy\cdot u^2\tau^2 = u^3\tau^2yi, 
    \end{cases}$ both of which are real. 
    Since  $u\in \mathbb{R} \setminus\{0\}$, we know $ \tau^2x, \tau^2yi \in \mathbb{R}$. 
    
    If $\tau = 0$, then $ y= xi$ and hence $z = 0 $ (because $\beta_2'$ is isotropic). 
    It follows that $\beta_2'$ can be absorbed into $\beta_1'$. 
    We can then rewrite $ Tf = \lambda\tent{((1,i,0)^T)}{3}$ for some $\lambda$. As $Tf$ is real, we know $ \lambda = 0$ and we are done.
    
    Now we can assume $\tau \ne 0$. Since both $\tau^2x, \tau^2yi \in \mathbb{R}$,  adding them  we know $\tau^3$ is real. Then for some $k \in \{0, 1, 2\}$,
    $\omega^k \tau \in \mathbb{R}$,
    where $\omega^3=1$. 
    Replacing $\beta_2'$
    by $\omega^k \beta_2'$, 
 which satisfies $\tent{(\omega^k \beta_2')}{3} = \tprime{\beta_2}{3}$,
    we may assume
    $\tau \in \mathbb{R} \setminus \{0\}$.

    We have $y = (x-\tau)i$. Since $\tau^2x \in \mathbb{R}$, we know $x \in \mathbb{R}$. 
    Because $Tf = \tprime{\beta_1}{3} + \tprime{\beta_2}{3} $ is real, we know $\begin{cases}
        (Tf)_{2,1,0} = u^3i + x^2 y = (u^3 + x^2(x-\tau))i\in \mathbb{R}\\
        (Tf)_{0,3,0} = y^3 - u^3i = -(u^3 + (x-\tau)^3)i\in \mathbb{R}.
    \end{cases}$ 
    From the fact that 
    $x,\tau \in \mathbb{R}$ and the above relations,
    we know $ x-\tau = - u$ (which gives $y=-ui$),
    and $x = \pm u$. 
    As $\beta_2'$ is isotropic, it follows that $ z=0$ and hence $\beta_2' = \pm u(1, \mp i, 0) ^T$. 
    If it is $+u(1, - i, 0) ^T$,
    then $\frac{1}{u^3}Tf = \tent{\beta_0}{3} + \tent{\overline{\beta_0}}{3}$.
    If it is $-u(1, + i, 0) ^T$,
    then $f = 0$.
    
\end{proof}

\begin{lemma} \label{rt3}
 If $f$ is a real-valued signature,  s.t., there exist $\beta, \gamma \in \mb{C}^3$,  $f = f_\beta + \tent{\beta}{2}\otimes\gamma + \beta \otimes \gamma \otimes \beta + \gamma \otimes \tent{\beta}{2}$, where $\beta \ne 0$, $\la\beta, \beta\ra = 0$, and $f_\beta$ is a complex ternary signature satisfying $\la f_\beta, \beta \ra = 0$,  then there exists a real orthogonal transformation $T$ such that $Tf = \lambda \tent{e_3}{3}$, $\lambda \in \mathbb{R}$. It implies that there is an $\alpha \in \mathbb{R}^3$, $f = \tent{\alpha}{3}$.
    
\end{lemma}

\begin{proof}
Let $\beta = (\beta_1, \beta_2, \beta_3)^T$, we have 
$\la f, \beta \ra = \la \gamma, \beta  \ra \tent{\beta}{2}$.
Let $M_i = \la f, e_i \ra$, for $i = 1,2,3 $, so $M_i$ is real symmetric.
Then $ M_i\beta = \la \la f, e_i \ra, \beta \ra = \la \la f, \beta\ra, e_i\ra = \beta_i\la \gamma , \beta \ra \beta$, for $i = 1,2,3$.   
As $\beta \ne 0$, we know $\beta_i\la \gamma, \beta \ra$ is a real eigenvalue of $M_i$.

Suppose $\la  \gamma, \beta \ra \ne 0$. Since $\beta \ne 0$, we can w.o.l.o.g. assume $\beta_1 \ne 0$. Then $ \frac{\beta_i}{\beta_1} = \frac{\beta_i\la \gamma, \beta \ra}{\beta_1\la \gamma, \beta \ra}$ is real and well defined ($i=1,2,3$). We can then write $\beta = \lambda u$, where $\lambda \in \mathbb{C}, u \in \mathbb{R}^3$ and $\lambda \ne 0, u \ne 0$. In particular,  $0 \ne  \lambda^2 \la u, u \ra = \la \beta, \beta \ra = 0$. 
This is a contradiction.

So we have $\la \gamma, \beta \ra  = 0$. Then,  $\la  f, \beta\ra = 0$. 
From Lemma~\ref{ct0}, we know there exists some real orthogonal matrix $T$ such that $T\beta = t\beta_0 $ where $t\in \mathbb{R}, \beta_0 = (1,i,0)^T$. Since $\beta \ne 0$, we have $t \ne 0$.
Then $0 =  \la f, \beta \ra = \la  T\beta, Tf\ra$. So $Tf$ has the form  
\begin{center} 
\begin{tabular}{c c c c c c c}
     &&& $a$ &&&  \\
     && $ai$ && $b$ &&\\
   & $-a$ && $bi$ && $c$ & \\
   $-ai$ && $-b$ && $ci$ && $d$ 
\end{tabular}    
\end{center}
Since $Tf$ is real, we get $a = b = c = 0$. Thus, $Tf = d\tent{e_3}{3}$ for some $d\in \mathbb{R}$.

\end{proof}

\begin{theorem} \label{co1}
Let $f$ be a 
\emph{real-valued} 
symmetric ternary function over domain $\{R,G,B\}$. Then $\Holant^*(f)$ is \#P-hard unless the function $f$ in expressible as one of the following two forms, in which case the problem is in FP. 
\begin{enumerate}
    \item $f = \tent{\alpha}{3} + \tent{\beta}{3} + \tent{\gamma}{3}$ where $\alpha, \beta, \gamma \in \mathbb{R}^3$ and $\la\alpha, \beta\ra = \la\beta,\gamma\ra = \la\gamma, \alpha\ra = 0$.
    \item $f = \tent{\alpha}{3} + \tent{\beta}{3} + \tent{\overline{\beta}}{3}$ where $\alpha \in \mathbb{R}^3$, $\la\alpha, \beta\ra = \la\beta,\beta\ra = 0$.
\end{enumerate}
This is equivalent to the
existence of a real orthogonal transformation $T$, s.t., \begin{enumerate}
    \item $Tf =  a\tent{e_1}{3} + b\tent{e_2}{3} + c\tent{e_3}{3}$ for some $a,b,c \in \mathbb{R}$.
    \item $cTf = \epsilon(\tent{\beta_0}{3} + \tent{\overline{\beta_0}}{3}) + \lambda\tent{e_3}{3}$ where $\beta_0 = (1,i,0)^T$, $\epsilon\in\{0,1\}$, and for some $c,\lambda \in \mathbb{R} $ and $c \ne 0$.
    
\end{enumerate}
\end{theorem}

\begin{proof}
    This follows from Theorems 3.1 and 3.2 in \cite{CaiLX13} and Lemmas~\ref{rt1}, \ref{rt2}, and~\ref{rt3}.
\end{proof}

Theorem~\ref{co1} is the adapted real dichotomy for $\Holant^*(f)$ with any real-valued signature $f$ 
of arity 3 on domain 3. 
We showed that in the real case,
we can take a real orthogonal transformation to the corresponding canonical forms. 
Also the third tractable case in the scenario of complex dichotomy does not exist anymore in the real case, which will simplify and expedite the analysis of further exploration of real dichotomies on domain 4.

\section{$\Holant^*(f)$   dichotomy for $\{0,1\}$-valued $f$ on domain 4}

A dichotomy for $\Holant^*(f)$ on domain 2 has been known~\cite{Fibo}.
With the real dichotomy for $\Holant^*(f)$ on domain 3, we 
will now explore the situation on domain 4. 
In this section we give a dichotomy for a single $\{0,1\}$-valued arity-3 signature $f$ taking input values on domain 4.


\begin{theorem} \label{d4tra}
Let $f$ be a \emph{real} symmetric  ternary function defined on a domain of size 4.
If there is a real orthogonal transformation $T$ such that $Tf$ has one of the following forms, then $\Holant^*(f)$   is P-time computable, where $\beta_0= (1,i,0,0)^T$ and $\beta_1= (0,0,1,i)^T$. 
\begin{enumerate}
    \item For some $a,b,c,d \in \mathbb{R}$, 
    $ Tf = ae_1^{\otimes 3} + be_2^{\otimes 3} + ce_3^{\otimes 3} + de_4^{\otimes 3} $.
    
    \item For some  $c, \lambda_1, \lambda_2 \in \mathbb{R}$,  and $c\ne 0$,
    $ cTf = \tent{\beta_0}{3} + \tent{\overline{\beta_0}}{3} + \lambda_1 \tent{e_3}{3} + \lambda_2 \tent{e_4}{3} $. 

    \item For some $\lambda_1, \lambda_2 \in \mb{R}$, 
    $ Tf = \lambda_1 (\tent{\beta_0}{3} + \tent{\overline{\beta_0}}{3}) + \lambda_2(\tent{\beta_1}{3} + \tent{\overline{\beta_1}}{3}) $.
\end{enumerate}
\end{theorem}
\begin{proof}
The proof is by a holographic transformation.
We omit the details.
\end{proof}

\begin{theorem} \label{main}
Let $f$ be a $\{0,1\}$-valued symmetric ternary function 
defined on a domain of size 4.
If $f$ is not among the P-time computable cases in  Theorem~\ref{d4tra}, then the problem $\Holant^*(f)$ is \#P-hard.
Moreover, for $\{0,1\}$-valued $f$,
only cases 1 and 2 in Theorem~\ref{d4tra}
occur.
%
%
%
%
%
\end{theorem}
We remark that for $\{-1, 1\}$-valued
 symmetric ternary functions, the
 third tractable case of Theorem~\ref{d4tra} does occur. 
The following is an example:

$\vcenter{\hbox{\begin{minipage}{3.5cm}
\centering
\includegraphics[width=4cm,height=4cm]{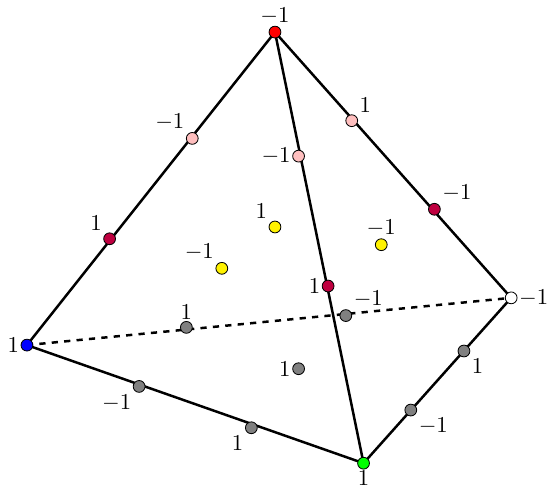}
\end{minipage}}}
\qquad
Q = \left[\begin{matrix}\frac{1 + \sqrt{2}}{2 \sqrt{\sqrt{2} + 2}} & - \frac{1}{2 \sqrt{\sqrt{2} + 2}} & \frac{- \sqrt{2} - 1}{2 \sqrt{\sqrt{2} + 2}} & \frac{1}{2 \sqrt{\sqrt{2} + 2}}\\\frac{-1 + \sqrt{2}}{2 \sqrt{2 - \sqrt{2}}} & \frac{1}{2 \sqrt{2 - \sqrt{2}}} & \frac{-1 + \sqrt{2}}{2 \sqrt{2 - \sqrt{2}}} & \frac{1}{2 \sqrt{2 - \sqrt{2}}}\\\frac{-3 + 2 \sqrt{2}}{2 \sqrt{10 - 7 \sqrt{2}}} & \frac{1 - \sqrt{2}}{2 \sqrt{10 - 7 \sqrt{2}}} & \frac{3 - 2 \sqrt{2}}{2 \sqrt{10 - 7 \sqrt{2}}} & \frac{-1 + \sqrt{2}}{2 \sqrt{10 - 7 \sqrt{2}}}\\\frac{-3 - 2 \sqrt{2}}{2 \sqrt{7 \sqrt{2} + 10}} & \frac{1 + \sqrt{2}}{2 \sqrt{7 \sqrt{2} + 10}} & \frac{-3 - 2 \sqrt{2}}{2 \sqrt{7 \sqrt{2} + 10}} & \frac{1 + \sqrt{2}}{2 \sqrt{7 \sqrt{2} + 10}}\end{matrix}\right]
$

On the left is a signature $g$. On the right is an orthogonal matrix $Q$. 
In fact, under the transformation, 
$Qg  = \sqrt{2-\sqrt{2}} (\tent{\beta_0}{3} + \tent{\overline{\beta_0}}{3}) - \sqrt{2+\sqrt{2}}(\tent{\beta_1}{3} + \tent{\overline{\beta_1}}{3})$, which is one example of the third tractable case of Theorem~\ref{d4tra}.

There are only a finite  (albeit a large) number of $\{0, 1\}$-valued
 symmetric ternary signatures on domain 4. 
We will prove Theorem~\ref{main}
by going through all signatures
using five general strategies.
When one signature could not be identified as \numP-hard by any of the five strategies in section~\ref{appro}, it is shown that it actually satisfies the first or second  tractable conditions in Theorem~\ref{d4tra}.

\subsection{Strategies} \label{appro}

There are  five different strategies we use to identify  \#P-hard signatures.

\begin{enumerate}
    \item Use gadgets to form a binary symmetric signature which when written as a matrix $M$ has rank 2.  Then apply an 
    orthogonal holographic transformation  $T$, which transforms $M$ to the form $\left[ \begin{smallmatrix} 0 & 0 & 0 & 0 \\ 0 & 0 & 0 & 0\\0 & 0 & \lambda & 0\\0 & 0 & 0 & \mu \end{smallmatrix}\right]$, 
    $\lambda\mu \ne 0$. We then interpolate ${\rm diag}(0,0,1,1)$, a Boolean equality on the last two (new) domain
    elements  by Lemma~\ref{eq3_2}.   Finally apply the Boolean domain dichotomy.

    \item Similarly, we form a binary symmetric signature which when written as a matrix has rank 3.  Then apply an orthogonal transformation $T$ and get $\left[ \begin{smallmatrix} 0 & 0 & 0 & 0 \\ 0 & \lambda_1 & 0 & 0\\0 & 0 & \lambda_2 & 0\\0 & 0 & 0 & \lambda_3 \end{smallmatrix}\right]$, 
    $\lambda_1 \lambda_2 \lambda_3 \ne 0$, and interpolate ${\rm diag}(0,1,1,1)$ by Lemma~\ref{eq3}. 
    Finally, apply 
    Theorem~\ref{co1} to $Tf$ on the last three (new) domain
    elements.

    \item Find  a nonzero unary signature $u \in \mathbb{R}^4$, such that   $\la f, u\ra  =0$. Then define an orthogonal matrix $T$ with  (normalized) $u$ as the first row.  
    $T$  transforms $f$ to a 
    signature supported on a lower domain (all 0's except the bottom face of the signature tetrahedron). Then apply the corresponding dichotomy.
    
    \item Find some nonzero unary signature $u \in \mathbb{R}^4$, 
    and nonzero $c \in \mathbb{R}$,
     such that $\la f, u \ra  = c u \cdot u^T$.  Define an orthogonal $T$ using (normalized) $u$  to be the first row. 
    $T$ will transform $f$ to be domain separated (where the first new domain element $R'$ is separated from the rest $\{G', B', W' \}$, i.e.,
    $Tf$ evaluates to 0, when $R'$ is among its input, \emph{except} possibly on $(R', R', R')$).  Then apply the domain 3 dichotomy Theorem~\ref{co1}. 
    
    \item Use gadgets to construct a symmetric binary  signature  $M$
     which when written as a
matrix has rank 4, and its 4 eigenvalues $\lambda_1, \lambda_2, \lambda_3, \lambda_4$ satisfy some condition under which we can interpolate (by Lemma~\ref{eq3_3} and Lemma~\ref{eq3_4}) either ${\rm diag}(0,0,1,1)$ or ${\rm diag}(0,1,1,1)$ from ${\rm diag}(\lambda_1, \lambda_2, \lambda_3, \lambda_4)$. 
    Form the orthogonal matrix $Q$ such that $Q M Q^T = {\rm diag}(\lambda_1, \lambda_2, \lambda_3, \lambda_4)$.
    Finally apply the corresponding lower domain dichotomy to the corresponding part of $Q f$.
\end{enumerate}
 
Below we show several 
examples using some of the  strategies above. 

Consider the tetrahedron  on the left. We use strategy  5:

$\vcenter{\hbox{\begin{minipage}{3.5cm}
\centering
\includegraphics[width=4cm,height=4cm]{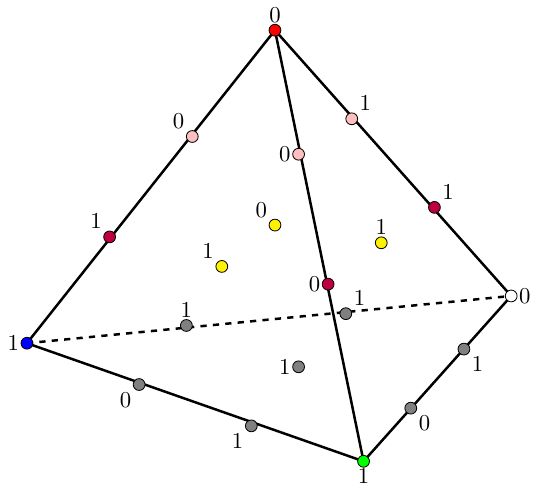}
\end{minipage}}}$
\qquad
\begin{tabular}{c c c c c c c}
     &&& $-\frac{41}{24\sqrt{3}}$ &&&  \\
     && $\frac{1}{4\sqrt{2}}$ && $-\frac{1}{12\sqrt{6}}$ &&\\
   & $\frac{2}{\sqrt{3}}$ && $0$ && $-\frac{5}{3\sqrt{3}}$ & \\
   $-\frac{2}{\sqrt{2}}$ && $\frac{\sqrt{2}}{\sqrt{3}}$ && $\frac{1}{\sqrt{2}}$ && $-\frac{\sqrt{2}}{3\sqrt{3}}$ \end{tabular}    

\noindent We call this signature $g$, and use a unary $e_4 = (0,0,0,1)^T$ to connect to $g$ to produce a binary symmetric function $M = \left[\begin{smallmatrix}1 & 1 & 0 & 1\\1 & 0 & 1 & 1\\0 & 1 & 1 & 1\\1 & 1 & 1 & 0\end{smallmatrix}\right]$ whose eigenvalues are $\{3,1,-1,-1\}$. 
We then use an orthogonal matrix $Q=\frac{1}{2\sqrt{3}}\left[\begin{smallmatrix}\sqrt{3} & \sqrt{3} & \sqrt{3} & \sqrt{3}\\-1 & -1 & -1 & 3\\- \sqrt{6} & 0 & \sqrt{6} & 0\\\sqrt{2} & - 2 \sqrt{2} & \sqrt{2} & 0\end{smallmatrix}\right]$ to diagonalize $M$, i.e., $QMQ^T = {\rm diag} (3,-1,1,-1)$. It follows that we are able to interpolate ${\rm diag}(0,1,1,1)$ (by Lemma~\ref{eq3_3}). We build a new signature grid where we add a  binary vertex 
on each edge between two $Qg$'s, and assign ${\rm diag}(0,1,1,1)$ on 
all the new degree 2 vertices.
This restricts all edges in the
new signature grid to be assigned a color only from $\{G',B',W'\}$ (no $R'$) in order the evaluation of any product term in the partition function to be nonzero. 
It follows that we have a  problem on domain size 3, which is 
defined by the ternary signature on domain 3 
shown on the right as a ``triangle". 
We can apply Theorem~\ref{co1}, and find that it is \numP-hard. Therefore, the problem
$\Holant^*(g)$  
is \numP-hard.

Next we show another example using strategy number 3 above.

$\vcenter{\hbox{\begin{minipage}{3.5cm}
\centering
\includegraphics[width=4cm,height=4cm]{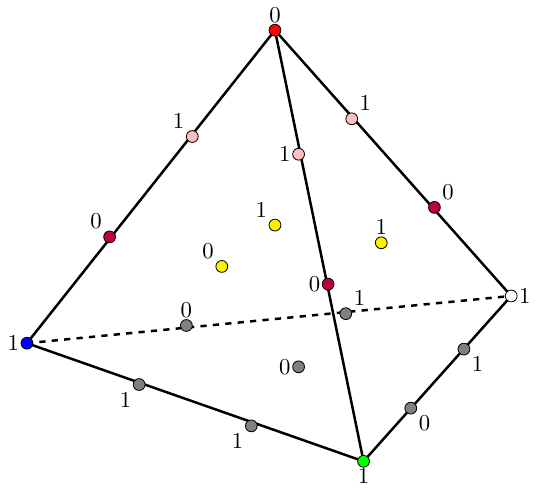}
\end{minipage}}}$
\qquad\qquad
$\vcenter{\hbox{\begin{minipage}{3.5cm}
\centering
\includegraphics[width=4cm,height=4cm]{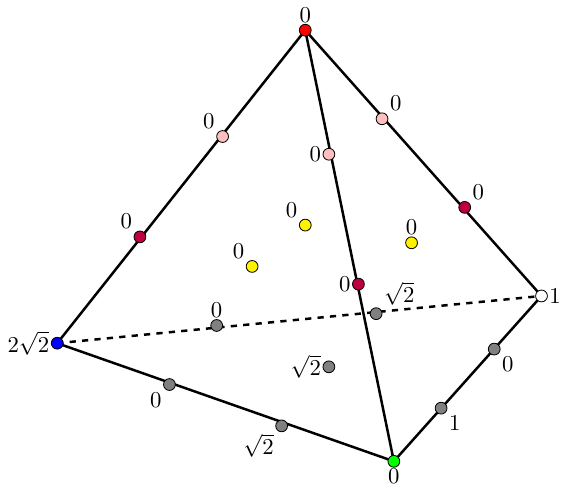}
\end{minipage}}}$

\noindent Let $g$ be the signature on  the left. Connecting a unary $u = (0,1,-1,0)^T$ to it, we get a zero binary function. We  construct an orthogonal matrix $Q$ with
 the normalized $u$ as its  first row,
$Q = \left[\begin{smallmatrix}0 & \frac{1}{\sqrt{2}} & - \frac{1}{\sqrt{2}} & 0\\1 & 0 & 0 & 0\\0 & \frac{1}{\sqrt{2}} & \frac{1}{\sqrt{2}} & 0\\0 & 0 & 0 & 1\end{smallmatrix}\right]$.
Then the transformed signature $Qg$ is shown as the tetrahedron on the right,  which contains only nonzero elements on its bottom face. 
This     signature is effectively a ternary signature  on domain size 3,
and  the original problem is equivalent to this problem on domain size 3.
Then we apply the domain 3 dichotomy.
By Theorem~\ref{co1}, 
the problem is \numP-hard.

Consider another example:

$\vcenter{\hbox{\begin{minipage}{3.5cm}
\centering
\includegraphics[width=4cm,height=4cm]{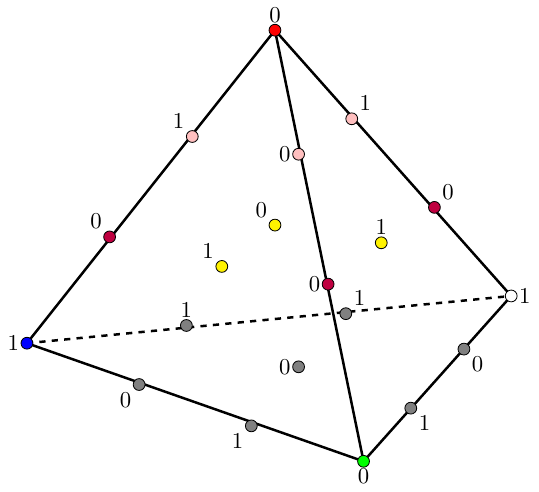}
\end{minipage}}}$
\qquad\qquad
$\vcenter{\hbox{\begin{minipage}{3.5cm}
\centering
\includegraphics[width=4cm,height=4cm]{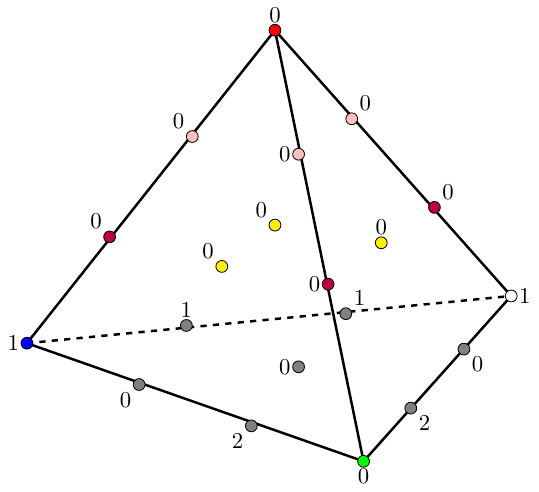}
\end{minipage}}}$

\noindent On the left is 
$g$,
a ternary signature on domain size 4.
There is a unary $u = (1,-1,0,0)^T$ such that $\la u,g\ra = 0$.
Hence, we  construct an orthogonal $Q=\left[\begin{smallmatrix}\frac{1}{\sqrt{2}} & - \frac{1}{\sqrt{2}} & 0 & 0\\\frac{1}{\sqrt{2}} & \frac{1}{\sqrt{2}} & 0 & 0\\0 & 0 & 1 & 0\\0 & 0 & 0 & 1\end{smallmatrix}\right]$ (where its first row is the normalized $u$), such that under the transformation, $g$ becomes $Qg$ (on the right) and 
$Qg=\tent{\alpha}{3} + \tent{\beta}{3}$
 where $\alpha = (0, -2^{\frac{1}{6}}, 2^{-\frac{1}{3}}, 2^{-\frac{1}{3}})^T$ and $\beta = (0, 2^{\frac{1}{6}}, 2^{-\frac{1}{3}}, 2^{-\frac{1}{3}})^T$. Here, after we use strategy number 3 to transform the problem to one on domain 3, we find it is tractable.

For all $\{0,1\}$-valued ternary signatures on domain size 4,  we went through them  using one of the above 5 strategies. It is found that we either can prove it
is \#P-hard, or when it fails to do so, 
it is tractable by being in 
one of the tractable forms in Theorem~\ref{main}.


\noindent
\paragraph{\bf Tricks in applying the domain 3 dichotomy} 
%
To apply the domain 3 dichotomy, there are also several tricks that can help simplify the calculation. 
\begin{enumerate}
    \item When checking whether a real ternary  domain 3 signature satisfies  the first tractable form $f= \tent{\alpha}{3} + \tent{\beta}{3} + \tent{\gamma}{3}$, we can use different unary signatures $u$ to connect to it and get different binary functions. Written as a matrix form $M = \la f, u\ra$, it's obviously symmetric.  
    Write its eigen-decomposition as $M = Q\Lambda Q^T$, where $Q$ is an orthogonal matrix $Q = [q_1,q_2, q_3]$ where $q_i$ ($i=1,2,3$) are column vectors,  and $\Lambda = {\rm diag}(\lambda_1, \lambda_2, \lambda_3)$. 
    If $\lambda_1, \lambda_2, \lambda_3$ are all distinct (at most one of them can be 0), and if $f$ falls into the first tractable case, then the set $\{\pm q_1, \pm q_2, \pm q_3\}$ is independent of $u$. 
    \begin{proof}
        If $f$ satisfies  the first tractable form $f= \tent{\alpha}{3} + \tent{\beta}{3} + \tent{\gamma}{3}$, then $M = \la u, f \ra = \la u, \alpha\ra \alpha \alpha^T + \la u, \beta\ra \beta \beta^T + \la u, \gamma \ra \gamma \gamma^T$. 
        Then $M \alpha = \la u, \alpha\ra \la \alpha, \alpha \ra \alpha$, so $\alpha$ is an eigenvector of $M$ (if nonzero). Similarly, nonzero $\beta, \gamma$ are both eigenvectors of $M$. 
        As ${\rm rank\,} M \ge 2$,
        at most one of  $\alpha$,  $\beta, \gamma$ can be 0.
        So at least two of them are
        scalar multiples of 
        $q_i$ $( i=1,2,3)$.
        So, at least two of
          $\pm q_i$ $( i=1,2,3)$
          do not depend on the choice  of $u$,
          and the third is uniquely
          determined by those two
          up to a $\pm$ factor.
     %
    %
    \end{proof}

    \item When checking whether a real ternary domain 3 signature is in the first tractability case, we can search for whether there exists a unary $u$ such that $\la u, f\ra = \la u, u \ra u u^T$. If a nonzero $f$  does not have  such  a nonzero unary, it cannot be in the first tractable case since  $\alpha, \beta, \gamma$  satisfy this relation.

    \item To check whether a real ternary domain 3 signature is in the second tractability case,  assume  it has  the form
    $f = \tent{\alpha}{3} + \tent{\beta}{3} + \tent{\overline{\beta}}{3}$. 
    If $\beta=0$, then it also falls into the first tractability category. 
    Now we assume $\beta \ne 0$. 
    Then we have a nonzero $\beta$ such that $\la \beta, f \ra = \la \beta, \overline{\beta} \ra \overline{\beta}\ \overline{\beta}^T$ and $\la \beta , \beta \ra =0$.
    So if there does not exist such a nonzero isotropic $\beta$, it is not in the second tractable case. 
\end{enumerate}

\subsection{Interpolate restricted equalities}
Suppose for some unary $u$, 
the 4-by-4 matrix $A = \innerp{u}{f}$ has   rank 3. 
Because $A$ is a real symmetric matrix, we can construct an orthogonal matrix $T$ 
so that
\begin{equation} \label{d4e3}
    TAT^T = \left[ \begin{smallmatrix} 0 & 0 & 0 & 0 \\ 0 & \lambda_1 & 0 & 0 \\ 0& 0 & \lambda_2 & 0 \\ 0 & 0 & 0 & \lambda_3 \end{smallmatrix}\right] 
\end{equation}
where $\lambda_1 \lambda_2 \lambda_3 \ne 0$.
We denote by $=_{G,B,W}$
the binary function in (\ref{d4e3})
if all $\lambda_i =1$.


\begin{lemma} \label{eq3}
Let $H:\{R,G,B,W\}^2 \to \mathbb{R}$ be a rank 3 binary function of the form (\ref{d4e3}). Then for any $\mathcal{F}$ containing $H$, we have $$ \Holant(\mf{F} \cup \{=_{G,B,W}\}) \le_T \Holant(\mf{F}).$$
\end{lemma}

Similarly, we denote by $=_{B,W}$
the binary function in (\ref{d4e3})
if  $\lambda_1 =0$ and $\lambda_2=\lambda_3 =1$. We have 

\begin{lemma} \label{eq3_2}
Let $H:\{R,G,B,W\}^2 \to \mathbb{R}$ be a rank 2 binary function of the form $\left[ \begin{smallmatrix} 0 & 0 & 0 & 0 \\ 0 & 0 & 0 & 0 \\ 0 & 0 & \lambda & 0 \\ 0 & 0 & 0 & \mu \end{smallmatrix}\right]$ where $\lambda \mu \ne 0$. Then for any signature set $\mf{F}$ containing $H$, we have $$ \Holant(\mf{F} \cup \{=_{B,W}\}) \le_T \Holant(\mf{F}). $$

\end{lemma}

Lemma~\ref{eq3} and \ref{eq3_2} 
enable us to construct instances on a lower domain which can help establish  \numP-hardness.
In the next two lemmas we interpolate a lower domain equality directly from a rank 4 real symmetric matrix.

\begin{lemma} \label{eq3_3}
Let $H:\{R,G,B,W\}^2 \to \mathbb{R}$ be a rank 4 symmetric binary function. Let $Q$ be the orthogonal matrix such that $Q H Q^T = \left[ \begin{smallmatrix} \lambda_1 & 0 & 0 & 0 \\ 0 & \lambda_2 & 0 & 0 \\ 0 & 0 & \lambda_3 & 0 \\ 0 & 0 & 0 & \lambda_4 \end{smallmatrix}\right]$ where $\lambda_1 \lambda_2 \lambda_3 \lambda_4 \ne 0$. If the four eigenvalues $\lambda_1, \lambda_2, \lambda_3, \lambda_4$ further satisfy the condition: for all $s, a,b,c \in \mathbb{Z}$, $s > 0$, if  $a + b + c = s$ then  $\lambda_1^s \ne \lambda_2^a \lambda_3^b \lambda_4^c$. Then for any $\mf{F}$ containing $H$, we have $$ \Holant(Q\mf{F} \cup \{=_{G,B,W}\}) \le_T \Holant(\mf{F}) .$$

\end{lemma}

Similarly, we have Lemma~\ref{eq3_4}.

\begin{lemma} \label{eq3_4}
Let $H:\{R,G,B,W\}^2 \to \mathbb{R}$ be a rank 4 symmetric binary function.  Let $Q$ be the orthogonal matrix such that $Q H Q^T = \left[ \begin{smallmatrix} \lambda_1 & 0 & 0 & 0 \\ 0 & \lambda_2 & 0 & 0 \\ 0 & 0 & \lambda_3 & 0 \\ 0 & 0 & 0 & \lambda_4 \end{smallmatrix}\right]$ where $\lambda_1 \lambda_2 \lambda_3 \lambda_4 \ne 0$. If the four eigenvalues $\lambda_1, \lambda_2, \lambda_3, \lambda_4$ further satisfy the condition: for all $s,t, a,b \in \mathbb{Z}$, 
$ s, t \ge 0, s+t>0$, if $a + b = s + t $ then $ \lambda_1^s \lambda_2^t \ne \lambda_3^a \lambda_4^b $. Then for any $\mf{F}$ containing $H$, we have $$ \Holant(Q \mf{F} \cup \{=_{B,W}\}) \le_T \Holant(\mf{F}) .$$

\end{lemma}



Let $\mathcal{U}$ be the set of all complex unaries and $f$ be a possibly complex signature, both over domain 3.
Then by definition,  $\Holant^*(f) =  \Holant(f \cup \mathcal{U})$. 
When $f$ is a real-valued signature, we temporarily define $\Holant^{r*}(f)$ as the $\Holant$ problem where only all the \textbf{real} unaries are available.
In Lemma~\ref{star}, we prove that in fact the complexity of the problem $\Holant^{r*}(f)$ is the same
as  $\Holant^{*}(f)$.
(Hence, after Lemma~\ref{star}, this new
notation $\Holant^{r*}(f)$ will be seen as unimportant, as far as the complexity of the problem  is concerned.)
 The same lemma  also holds for any domain $k$; the proof of Lemma~\ref{star} is easily adapted.

 \begin{lemma} \label{star}
 For any real signature $f$ over domain 3, we have 
  $$ \Holant^*(f) \le_T \Holant^{r*}(f). $$
 \end{lemma}

\begin{proof}
    By Theorem 3.1 in~\cite{CaiLX13}, the problem $\Holant^*(f)$ is either polynomial-time solvable, or it is \#P-hard. If it is polynomial-time solvable, then the reduction trivially holds. Otherwise, it is \#P-hard. This means  that 
    there exist finitely many unaries $u_1, u_2, \dots, u_k$, possibly complex-valued, such that the problem $\Holant(\{f, u_1, \dots, u_k\}) $ is \#P-hard. 
    We
   prove that  $\Holant(\{f, u\}) \le_T \Holant^{r*}({f})$,
   for any complex-valued  $u$, and then use induction,
   as $k$ is a constant. 

    Now consider an arbitrary instance $I$ of $\Holant( \{f, u\})$ and suppose the  unary 
    function $u$ appears $m$ times in $I$. We write $u = [x,y,z]$ where $x,y,z \in \mathbb{C}$. 
    We may assume $u \ne 0$, and by symmetry we may assume $z \ne 0$. We stratify the edge assignments according to the number of Red, Green and Blue assigned to the $u$'s. Specifically, let $\rho_{ij}$ denote the sum of products of
    the evaluations of  all $f$'s, where the sum is over all assignments  with exactly
    $i$ many times Red, $j$ many times Green, and $(m-i-j)$ many times Blue are assigned to the $u$'s. Then the Holant value on the instance $I$ can be written as $\sum\limits_{i=0}^m \sum\limits_{j=0}^{m-i} \rho_{ij}x^i y^j z^{m-i-j} = z^m\sum\limits_{i=0}^m \sum\limits_{j=0}^{m-i} \rho_{ij}(\frac{x}{z})^i (\frac{y}{z})^j.$ Observe that once we know all the values of $\rho_{ij}$, we can compute the Holant value of $I$ in P-time. We now interpolate the values of $\rho_{ij}$.

    We construct the instances $I_k$, $1\le k \le \binom{m+2}{2}$, for the problem $\Holant^{r*}({f})$. For any such $k$, let $I_k$ be the same signature grid as $I$ except we replace every appearance of $u$ by the real unary $[3^k,2^k, 1]$. Then, the Holant value for $I_k$ equals to $\sum\limits_{i=0}^m \sum\limits_{j=0}^{m-i} \rho_{ij}3^{ki} 2^{kj}$. Therefore, we can write a non-degenerate Vandermonde system, where the matrix has entries $3^{ki}2^{kj}$ with the columns indexed by lexicographical order of the tuples $(i,j)$ and the rows indexed by $k$, and the unknown variables are $\rho_{ij}$. We can then solve for $\rho_{ij}$ and therefore compute the Holant value of $I$ in polynomial time. This finishes the proof of the lemma.
    
\end{proof}

%
%
%
\bibliographystyle{splncs04}
\bibliography{main}
%





\newpage
\section*{Appendix}

We start with the proof of Lemma~\ref{eq3}.

\begin{proof}

Fix an instance $I$  of $\Holant(\mf{F} \cup\{=_{G,B,W}\})$. Suppose the signature $=_{G,B,W}$ appears $m$ times in $I$. 
 To have a nonzero contribution to the Holant sum, the assignments given to any occurrence of the $=_{G,B,W}$ signatures must be  $(G,G) $, $(B,B)$ or $(W,W)$.
 We can stratify the edge assignments in the Holant sum of $I$ according to the number of $(G,G)$, $(B,B)$ and $(W,W)$  assigned to the occurrences of $=_{G,B,W}$.
 Specifically, let $\rho_{ij}$ denote the sum of products of the evaluations of all other signatures, 
 where the sum is over all assignments with exactly $i$ many times  $(G,G)$, $j$ many times $(B,B)$ and $(m-i-j)$ many times $(W,W)$ assigned to those $=_{G,B,W}$ signatures. 
 Then the Holant value of the instance $I$  can be written as $\sum_{i=0}^{m}\sum_{j=0}^{m-i}\rho_{ij}$.
 
 Let $K = \binom{m+2}{2}$. Now we construct from $I$ a sequence of instances $I_k$ of $\Holant(\mf{F})$ indexed by $k$, $1 \le k \le K$, by replacing each occurrence of $(=_{G,B,W})$ in $I$
 with a chain of $k$ copies of the function $H$. 
 The signature of the $k$-chain is 
 $H^k  = \left[ \begin{smallmatrix} 0 & 0 & 0 & 0 \\ 0 &\lambda_1^k & 0 & 0 \\ 0 &0 & \lambda_2^k & 0 \\ 0 & 0 & 0 & \lambda_3^k \end{smallmatrix}\right]$.
 This effectively replaces each
  occurrence of $(=_{G,B,W})$ in $I$
 by $H^k$.
 Let $x_{ij}$ be $(\frac{\lambda_1}{\lambda_3})^{i} (\frac{\lambda_2}{\lambda_3})^{j}$. 
 A moment of reflection shows that the value of the instance $I_k$ is $$\sum\limits_{i=0}^{m}\sum\limits_{j=0}^{m-i} \rho_{ij} \lambda_1^{ki}\lambda_2^{kj}\lambda_3^{k(m-i-j)}  = \lambda_3^{mk} \sum\limits_{i=0}^{m}\sum\limits_{j=0}^{m-i} \rho_{ij} 
 x_{ij}^k
 . $$
 
 Then we can view the above as a linear system where the coefficients are 
 $x_{ij}^k$
 and the variables are $\rho_{ij}$. 
 Two  columns indexed by $(i,j) \ne (i',j')$ are identical in the matrix
 $
\begin{bmatrix}
    1 & \cdots & x_{ij} & \cdots & x_{i' j'} & \cdots \\
    1 & \cdots & x_{ij}^2 & \cdots & 
    x_{i' j'}^2 & \cdots \\
    \vdots & \ddots & \vdots & \ddots &\vdots & \ddots \\
    1 & \cdots & x_{ij}^K & \cdots & x_{i' j'}^K & \cdots
\end{bmatrix}$, 
iff  $x_{ij} = x_{i' j'}$.
%
 We can combine these two columns and create a new variable representing the sum of $\rho_{ij}$ and $ \rho_{i'j'}$. Repeating in this way until no two columns are the same, we will arrive at a 
 matrix of $K' \le K$ columns
 such that the first $K'$ rows form a  square Vandermonde matrix with full rank. We can then solve the Vandermonde system and take the
 sum of all the  variables. This is the sum of the original $\rho_{ij}$'s, which is exactly the Holant value on $I'$.
 The proof is now complete.

\end{proof}

The proof of Lemma~\ref{eq3_2} can be easily adapted from the proof of Lemma~\ref{eq3}. We now continue with the proof of Lemma~\ref{eq3_3}.

\begin{proof}
We have $$  \Holant(Q\mf{F}) \equiv_T \Holant(\mf{F}), $$
and the binary function $Q H Q^T=\left[ \begin{smallmatrix} \lambda_1 & 0 & 0 & 0 \\ 0 & \lambda_2 & 0 & 0 \\ 0 & 0 & \lambda_3 & 0 \\ 0 & 0 & 0 & \lambda_4 \end{smallmatrix}\right] \in Q\mf{F}$ because $H\in \mf{F}$. 

 Consider an instance $I$ of $\Holant(Q\mf{F} \cup \{=_{G,B,W}\})$. Suppose the function $(=_{G,B,W}) = \left[ \begin{smallmatrix} 0 & 0 & 0 & 0 \\ 0 & 1 & 0 & 0 \\ 0 & 0 & 1 & 0 \\ 0 & 0 & 0 & 1 \end{smallmatrix}\right]$ appears $m$ times in $I$. 
 Similar to the proof of Lemma~\ref{eq3}, 
 we can stratify the edge assignments in the Holant sum according to the number of $(R,R)$, $(G,G)$, $(B,B)$ and $(W,W)$  assigned to the occurrences of $(=_{G,B,W})$. Specifically, 
 let $\rho_{\ell ij}$ denote the sum of products of the evaluations of all other signatures, where the sum is over all assignments with exactly $\ell$ many times $(R,R)$, $i$ many times  $(G,G)$, $j$ many times $(B,B)$ and $(m-\ell-i-j)$ many times $(W,W)$ assigned to those $(=_{G,B,W})$ signatures.
  Then, the Holant value on the instance $I$  can be written as $\sum_{i=0}^{m}\sum_{j=0}^{m-i}\rho_{0ij}$, since any edge assignment that makes nonzero contribution in the Holant sum must assign 0 time $(R,R)$ to $(=_{G,B,W})$.
 
 Now we construct from $I$ a sequence of instances $I_k$ for $\Holant(Q\mf{F})$ indexed by $k$, $1 \le k \le \binom{m+3}{3}$, by replacing each occurrence of $(=_{G,B,W})$ with a chain of $k$ copies of the function $\left[ \begin{smallmatrix} \lambda_1 & 0 & 0 & 0 \\ 0 & \lambda_2 & 0 & 0 \\ 0 & 0 & \lambda_3 & 0 \\ 0 & 0 & 0 & \lambda_4 \end{smallmatrix}\right]$. 
That is, each occurrence of $\left[ \begin{smallmatrix} 0 & 0 & 0 & 0 \\ 0 & 1 & 0 & 0 \\ 0 & 0 & 1 & 0 \\ 0 & 0 & 0 & 1 \end{smallmatrix}\right]$ is replaced by $\left[ \begin{smallmatrix} \lambda_1^k & 0 & 0 & 0 \\ 0 & \lambda_2^k & 0 & 0 \\ 0 & 0 & \lambda_3^k & 0 \\ 0 & 0 & 0 & \lambda_4^k \end{smallmatrix}\right]$.
 Let $x_{\ell ij}$ be $(\frac{\lambda_1}{\lambda_4})^{\ell }(\frac{\lambda_2}{\lambda_4})^{i}(\frac{\lambda_3}{\lambda_4})^{j}$.
 A moment of reflection shows that the Holant value of the instance $I_k$ is $$\sum\limits_{\ell=0}^{m}\sum\limits_{i=0}^{m-\ell}\sum\limits_{j=0}^{m-\ell-i} \rho_{\ell ij} \lambda_1^{\ell k}\lambda_2^{ik}\lambda_3^{jk}\lambda_4^{(m-\ell-i-j)k} = \lambda_4^{mk}\sum\limits_{\ell=0}^{m}\sum\limits_{i=0}^{m-\ell}\sum\limits_{j=0}^{m-\ell-i} \rho_{\ell ij} x_{\ell ij}^k.$$
 
 Then we can view the above as a linear system where the coefficient matrix has entries $x_{\ell ij}^k$ with the columns indexed by lexicographical order of the tuples $(\ell,i,j)$ and the rows indexed by $k \in \mathbb{Z}_+$, and the unknown variables are $\rho_{\ell ij}$.  We will take $k =1,\ldots, K$,
 where $K = \binom{m+3}{3}$.
 We combine two columns 
 indexed by $(\ell, i, j) \ne (\ell', i', j')$ 
 if the two columns  are equal. This happens iff  $\lambda_1^{\ell}\lambda_2^{i}\lambda_3^{j}\lambda_4^{m-\ell-i-j} = \lambda_1^{\ell'}\lambda_2^{i'}\lambda_3^{j'}\lambda_4^{m-\ell'-i'-j'}$. In that case  we create a new  variable $\rho'$ to be the sum of the two variables $ \rho_{\ell ij}$ and $ \rho_{\ell'i'j'}$. 
 We keep combining the columns in this way until no two columns are the same, at which point we arrive at a Vandermonde matrix with full rank.
 Notice that what we want is $\Holant(I) = \sum_{i=0}^{m}\sum_{j=0}^{m-i}\rho_{0ij}$. 
 Therefore, as long as
 $(0, i, j)$ and $ (\ell', i', j')$
 for $\ell' >0$ are never combined,
 the sum of $\rho_{0ij}$'s and the sum of $\rho_{\ell'i'j'}$'s are not ``mixed''.  This condition is
 expressed as:
For all $i,j,i',j' \ge 0$,
and $\ell' > 0$, 
$\lambda_2^{i}\lambda_3^{j}\lambda_4^{m-i-j} \ne \lambda_1^{\ell'}\lambda_2^{i'}\lambda_3^{j'}\lambda_4^{m-\ell'-i'-j'}$. This is guaranteed by the assumption of the lemma. 
 
 Therefore, we can solve the Vandermonde system, 
  summing up the ones coming from the form $\rho_{0ij}$ where $i+j \le m$, to get $\Holant(I)$.
 Therefore, 
 $\Holant(Q\mf{F} \cup \{=_{G,B,W}\}) \le_T \Holant(\mf{F})$ and the lemma is proved.

\end{proof}

The proof of Lemma~\ref{eq3_4} can be easily adapted from the proof of Lemma~\ref{eq3_3}. We will define  
$\rho_{\ell i j}$'s  similarly.
We want to compute the Holant value $\sum_{j=0}^{m}\rho_{00j}$.
The sum of $\rho_{00j}$'s and the sum of $\rho_{\ell' i' j'}$'s are not ``mixed'' as long as for all
$\ell', i'\ge 0$, with $\ell'+i' >0$, and all $j, j' \geq 0$, 
we have $\lambda_3^j \lambda_4^{m-j} \neq \lambda_1^{\ell'} \lambda_2^{i'} \lambda_3^{j'} \lambda_4^{m-\ell'-i'-j'}$, which is guaranteed by the assumption of Lemma~\ref{eq3_4}.

\end{document}